\documentclass[11pt,letterpaper]{article}
\usepackage{dsfont,amsmath,amsthm,amssymb}
\usepackage{subfigure}
\usepackage{tikz}
\usetikzlibrary{shapes.symbols,arrows}
\usepackage{enumerate}
\usepackage{fullpage}

\newtheorem{theorem}{Theorem}
\newtheorem{proposition}[theorem]{Proposition}
\newtheorem{lemma}[theorem]{Lemma}
\newtheorem{corollary}[theorem]{Corollary}

\newcommand{\IR}{\mathds{R}}
\newcommand{\IN}{\mathds{N}}

\renewcommand{\leq}{\leqslant}

\renewcommand{\geq}{\geqslant}

\newcommand{\one}{\mathds{1}}

\begin{document}

\title{Orientation and Connectivity Based Criteria \\ for Asymptotic Consensus}

\author{Bernadette Charron-Bost\textsuperscript{1} 
}

\date{
\textsuperscript{1} CNRS, LIX, \'Ecole polytechnique, 91128 Palaiseau, France\\
}

\maketitle

\begin{abstract}%
In this article, we establish orientation and connectivity based criteria for the agreement algorithm 
	to achieve asymptotic consensus in the context of time-varying topology and
	communication delays.
These criteria unify and extend many earlier convergence results
	on the agreement algorithm for deterministic and discrete-time multiagent systems.
\end{abstract}

\section{Asymptotic consensus in a multiagent system}\label{sec:intro}

Let us consider a set of autonomous agents that interact with each other by
	exchanging values and  perform instantaneous operations on received
	values.
Agents each  start with a real value and must reach agreement on a value which is
	a convex combination of the initial values.
The agents are not required to agree exactly as in the decision problem called consensus
 	in fault-tolerance~\cite{Lyn96}, but ought to
	iteratively compute values that all converge to the same limit.
	
The motivation for this {\em asymptotic consensus problem} comes from a variety of 
	contexts involving distributed systems.
For example, sensors (altimeters) on board an aircraft could be attempting to reach
	agreement about the altitude.
Or a collection of clocks that are possibly drifting apart have to maintain clock values
	that are sufficiently close.

For the multiagent systems described above, the asymptotic consensus 
	problem has been proposed  a solution	which consists in an iterative 
		linear procedure, classically referred to as the {\em agreement algorithm}. 
It has been introduced by DeGroot~\cite{DG74} for the synchronous and time-invariant case, 
	and then has been extended  by Tsitsiklis et al.~\cite{Tsi84,TBA86} to the case of asynchronous 
	communications and time-varying  environment.
A related algorithm has been later proposed by Vicsek et al.~\cite{VCBCS95} as a model of 
	cooperative behavior.
The subject has recently attracted considerable  interest within the context of flocking
	and multiagent coordination (see for instance~\cite{JLM03,Mor05,Cha12,OSFM07} 
	for surveys and references).
	
In this article, we establish orientation and connectivity based criteria for the agreement algorithm 
	to achieve asymptotic consensus in the context of time-varying topology and
	communication delays.
These criteria unify and extend  earlier convergence results, namely the one  
		in~\cite{Tsi84, TBA86,Mor05,HB05,BHOT05,CSM05,TN11, HT11},
		and notably concern the {\em coordinated} and the {\em decentralized} models 
 		of multiagent systems that we define by simple {\em orientation} and
 		{\em connectivity} properties on their communication graphs.

Our proofs of convergence rely on uniform techniques.
They  share the same core, and only differ in the control of the convergence speed: 
	in coordinated systems,
 	the convergence speed depends quadratically on the number of agents while
	it is finite but unbounded in decentralized sytems. 

\subsection{The agreement algorithm with time-varying topology and communication delays}\label{sec:agreement}

We briefly recall the model for the agreement algorithm, and the set of assumptions 
	that are usually made.
We consider a set of $N$ agents denoted $1,\cdots,N$.
We assume the existence of a discrete global clock and we take the range
	of the clock's ticks to be the set $\IN$ of natural numbers.
The state of the agent~$i$ is captured by a scalar variable $x_i$,
	and the  value held by~$i$ at time~$t$ is denoted $x_i(t)$.
Each agent~$i$ starts with an initial value~$x_i(0)$, and the evolution
	of the local variable~$x_i$ is described by the linear transition function:
	\begin{equation}\label{eq:linear}
		x_i(t+1) = \sum_{j=1}^N A_{i,j}(t) x_j(\tau_{ij}(t)) \enspace.
	\end{equation}
Equation~(\ref{eq:linear}) corresponds to the fact that at each time~$t+1$,
	the agent~$i$ updates $x_i$ with a weighted average of the values 
	it has received at time $t$.
In the presence of communication delays, the value received by~$i$ from~$j$ 
	at time~$t$ may be an outdated value, i.e.,  sent by~$j$ at some time~$\tau_{ij}(t)$
	with $0\leq \tau_{ij}(t) \leq t$.
For each time~$t$, we form the $N \times N$ matrix~$A(t)$ and the {\em communication
	graph} $G(t) = \big( [N]\,, E(t)\big )$ which is the directed graph with a node 
	for each agent in $[N]= \{1,\cdots,N\}$ and where there is an edge from~$i$ to~$j$ 
	if and only if $A_{i,j}(t) >0$.
In other words, the agent~$i$ is connected to the agent~$j$ in $G(t)$ if~$i$ hears
	of~$j$ at time $t$.

We now formulate a series of assumptions on the matrices $A(t)$ and the delays $\tau_{ij}(t)$,
	which hold naturally in the context of a 
	multiagent system  running the agreement algorithm.
\begin{description}
	\item[A1:] Each matrix $A(t)$ is stochastic.
	\item[A2:] Each communication graph $G(t)$ contains all possible self-loops, i.e.,
		        $A_{i,i}(t) >0$ for all $i \in [N]$.
	\item[A3:] The positive entries of the matrices $A(t)$, $t\in \IN$, are uniformly lower
		bounded, i.e., there exists $\alpha \in ]0,1]$ such that $A_{i,j}(t) \in \{0\}\cup [\alpha,1]$
		for all $i,j \in [N]$ and all $t\in \IN$.
\end{description}

Concerning the delays, we assume
\begin{description}
	\item[B1:] $\tau_{ij}(t)  \leq t$, for all $i,j \in [N]$ and all $t \in \IN$.
	\item[B2:] $\tau_{ii}(t) =t$, for all $i \in [N]$ and all $t \in \IN$.
	\item[B3:] There exist  some  positive integer $\Delta$ such that 
		        $\tau_{ij}(t) \geq  \max(0, t-\Delta +1)$ for all $i,j \in [N]$
				and all $t \in \IN$.
\end{description}
Assumption A1 corresponds to the updating rules of the $x_i$'s in terms of weighted averages
	discussed above.
Assumptions A2 and B2 express the fact that an agent has an immediate access 
	to its own current value.
Assumption A3 is obviously fulfilled when the set of matrices $A(t)$ is finite.
Assumption B1 captures the fact communication does not violate causality:
	a future value of agent~$j$ cannot influence the computation of 
	agent~$i$'s value.
Finally, with assumption B3  we suppose the multiagent system to be partially synchronous, namely
	communication delays are bounded.
However since we do not require the functions $\tau_{ij}$ to be either non-decreasing, 
	surjective, or injective, communications between agents may be non-FIFO and unreliable 
	(duplication and loss). 
The case of zero communication delays is captured by assumptions A and B with $\Delta =1$, 
	and  equation~(\ref{eq:linear})
	corresponds in this case to the evolution of the agreement algorithm for a synchronous multiagent
	system.
	
\subsection{The coordinated and decentralized models}\label{sec:CD}
	
All the above assumptions are classical, contrary to the conditions C and D  we introduce now,
	that lie at the heart of our main convergence theorems.
	
Let us recall that for a directed graph~$G$, its  {\em strongly connected components}
	are, in general,  strictly included into  its {\em connected components} defined as  the 
	connected components of the undirected version of~$G$.
Let us also recall that the directed graph~$G=(V,E)$ is said to be $j$-{\em oriented}, for $j\in V$,
 	if for every node there exists a directed path originating at this node and
	terminating at~$j$~\cite{GB81}. 
If~$G$ is $j$-oriented for some node~$j$, then $G$ is said to be {\em oriented}.
We now introduce the following condition on sequences of communication graphs.
\begin{description}
	\item[ C :] At every time $t\in\IN$, the communication graph $G(t)$ is oriented.
\end{description}
Intuitively, while the communication graph is $j$-oriented, the agent~$j$ gathers the 
	values in its strongly connected component, computes  some average value,
	and attempts to impose this value to the rest of the agents as long as
	the communication graph remains $j$-oriented.
In other words, its particular position in the communication graph makes~$j$ 
	to play the role of {\em system coordinator} for the agreement algorithm.
Accordingly, we define the {\em coordinated model} as the model of multiagent systems
	which, in addition to assumptions A and B, satisfy condition C.
	
From the above discussion about the role of coordinator, it is 	
	easy to grasp why in the particular case of a steady coordinator,
	all the agents converge to a common value when running the agreement algorithm.
Our first  theorem shows that asymptotic consensus is actually achieved even when coordinators
	change over time. 

\begin{theorem}\label{thm:coord}
In the coordinated model, the agreement algorithm guarantees asymptotic consensus.
\end{theorem}

We introduce a second model of multiagent systems, the {\em 
	decentralized model}, in which the orientation  condition C of the coordinated model 
	is replaced by  two connectivity conditions~D1 and~D2.
Before stating them, let us recall that  a directed graph is said to be {\em completely reducible}
	 if all its connected components are  strongly connected. 
\begin{description}
	\item[D1:] For every time $t\in \IN$, the directed graph $([N]\,, \cup_{s\geq t}E(s))$ is strongly 
	connected. 
	\item[D2:] At every time $t\in \IN$, the communication graph $G(t)$ is completely reducible.
\end{description}

The second main result of this paper is the following convergence theorem for the agreement algorithm.

\begin{theorem}\label{thm:decent}
In  the decentralized model, the agreement algorithm guarantees asymptotic consensus.
\end{theorem}

The convergence mechanism behind this result can be understood as follows in 
	intuitive terms: there is no source (that is, no node without incoming edge) in a 
	completely reducible directed graph, and thus for the decentralized model,
	there is no dead end in the information flow
	of each connected component of the communication graph.
The strong connectivity condition~D1 guarantees that the
	values computed in each connected component are then spread out over
	the whole system.
Even if at a given time, the roles performed by agents may be not at all equivalent
 	since the communication graph may be non-symmetric, all the agents
 	eventually play the same role over time and converge to the same value. 

Because of the self-loop assumption A2, the decentralized model
	corresponds to  a weak form of ergodicity, namely 
	each matrix $A(t)$ is block ergodic.
Similarly, a close inspection of our proof of Theorem~\ref{thm:coord} 
	reveals that in the coordinated model, each matrix $A(t)$ is ``partially ergodic'' 
	in the sense that there exist some index $j$ in $[N]$ 
	and some positive integer $n$ such that all the entries in the $j$-th column of 
	the matrix $\big (A(t)\big)^n$ are positive.
Clearly a matrix that is both block ergodic and partially ergodic is ergodic.
Since a strongly connected directed graph is oriented with respect to each of
	its nodes, the intersection of the coordinated model and the decentralized
	model indeed coincides with the model where communication graphs are all 
	strongly connected.

\subsection{Some strengthenings of our convergence theorems}\label{sec:generaltheorem}

Theorems~\ref{thm:coord} and~\ref{thm:decent} admit generalizations --- which turn out to be useful 
 	in applications --- where conditions~C and~D
	are weakened in diverse directions while maintaining convergence.
	
Firstly, convergence is maintained when condition~C (resp. D) holds only eventually: C (resp. D)
	may be violated during a finite period, but is supposed to hold from  some time onward.
Roughly speaking, that corresponds to a realistic situation where the multiagent system
	stabilizes after some transient phase during which the communication graph is arbitrary, 
	of duration unknown to the agents.
	
Secondly, we may address the issue of the {\em granularity} at which condition~C (resp. D)
	holds: instead of observing  the multiagent system at each time $t\in \IN$, we might 
	have access to its state only at the end of each period of time of a fixed duration $\Phi$.
Formally, that corresponds to the introduction of a new time scale $u\in \IN$, and to let
	$t = \Phi u$.
In the synchronous case where we assume $\tau_{ij}(t)=t$, the evolution equation
	(\ref{eq:linear}) takes the form 
	\begin{equation}\label{eq:granularity}
	x_i\big( \Phi(u+1)\big) = \sum_{j=1}^N \tilde{A}_{i,j}(u) x_j(\Phi u) 
	\end{equation}
	where $\tilde{A} = A(\Phi u + \Phi - 1) \cdots A(\Phi u)$.
In other words, in the synchronous model, the change of granularity amounts to
	grouping the matrices in blocks of length~$\Phi$, and in replacing
 	each block by the product of the matrices in the block.
It is remarkable that, in the general non synchronous case captured by assumption B, where
	the maximum delay $\Delta$ is any positive integer and where the system
	evolution does {\em not} follow equation (\ref{eq:granularity}) anymore,
	asymptotic consensus is still guaranteed under a simple extension of 
	condition C involving  matrix products of length $\Phi$.
	
Actually, these two types of weakening of conditions C and D can be combined in  
	the following formulations, respectively.
\begin{description}
	\item[$\Diamond$C:] There exist a time $T_0\in \IN$ and a positive integer $\Phi$
	such that at every time $t \geq T_0$, the communication graph~$H(t)$ of the product 
	$A(t+\Phi -1) \cdots A(t)$ is oriented.
	\item[$\Diamond$D2:] There exist a time $T_0\in \IN$ and a positive integer $\Phi$
	such that at every time $t \geq T_0$, the communication graph~$H(t)$ of the product 
	$A(t+\Phi -1) \cdots A(t)$ is completely reducible.
\end{description}

Simple extensions of our proofs  will allow us to show 
	the following generalizations of Theorems~\ref{thm:coord} 
	and~\ref{thm:decent}.

\begin{theorem}\label{thm:diamondC}
Under assumptions A and B,	the agreement algorithm guarantees asymptotic consensus
	when condition $\Diamond$C holds.
\end{theorem}

\begin{theorem}\label{thm:diamondD}
Under assumptions A and B,	the agreement algorithm guarantees asymptotic consensus
	when condition $\Diamond$D, the conjunction of D1 and $\Diamond$D2, holds.
\end{theorem}

Beside the  above variants of conditions C and D (derived from~C and~D by a weakening
	procedure standard in temporal logic), a 
	close inspection of the proof of Theorem~\ref{thm:decent} leads us to introduce
	another weakening of condition~D.
Indeed, 
	our proof shows that the agreement algorithm achieves asymptotic 
	consensus when replacing condition D by the following weaker condition.
\begin{description}
	\item[ D* :] There is some $j \in [N]$ such that at every time $t\in\IN$, 
	\begin{enumerate}
		\item the directed graph $([N]\,, \cup_{s\geq t}E(s))$ is $j$-oriented;
		\item the connected component of $j$  in the communication graph $G(t)$ is $j$-oriented, and 
		every other connected component of~$G(t)$  is strongly connected.
\end{enumerate}
\end{description}

\begin{theorem}\label{thm:decent*}
Under assumptions A and B,	the agreement algorithm guarantees asymptotic consensus
	when condition D* holds.
\end{theorem}

Contrary to the coordinators in condition C, the agent~$j$ in condition D2* is fixed 
	in time.
Indeed, the combination of conditions C and D2* defines a simple model of a {\em steady 
	coordinator}. 
One point of interest in Theorem~\ref{thm:decent*} is that it demonstrates that 
	asymptotic consensus can  be achieved in a hybrid model: agents can
	be disconnected from the coordinator provided  they are still
	clustered into independent and strongly connected groups.

\subsection{Related work}

Numerous convergence results for the agreement algorithm  have been established in the literature.

Presumely, the first one, which  is a straightforward corollary of the classical Frobenius' theorem~\cite{Fro12},
	concerns the case of a {\em fixed ergodic} matrix  in the synchronous setting.
Wolfowitz's theorem~\cite{Wol63} extends this result to the  sequences of varying matrices 
	taken from a finite set of ergodic matrices such that any finite product of matrices
	in that set is ergodic.
We refer the reader to~\cite{Sen73} for historical references and variants
	of these theorems.
	
Bertsekas and Tsitsiklis~\cite{BT89} introduced the set of assumptions A and B to relax the 
	finiteness hypothesis on the set of matrices and to handle communication delays.
Moreover they defined a condition on the sequence of communication
	graphs, the condition of  $\Phi$-{\em bounded intercommunication intervals}, where 
	$\Phi$ denotes a  positive integer: if $(i,j)$ is an edge of the 
	communication graph infinitely often, then  $(i,j)$ is required to be
	an edge of the communication graph at least once during each period of duration $\Phi$.
Tsitsiklis~\cite{Tsi84} proved that under assumptions A, B, D1 and on the condition of bounded 
	intercommunication intervals, the agreement algorithm guarantees asymptotic consensus.
It is easy to see that assumptions A2 and D1 combined with the condition 
	of $\Phi$-bounded intercommunication intervals imply that from  some time onward,
	any product of $\Phi N$ consecutive 
	matrices in the sequence $(A(t))_{t\in \IN}$ is a positive matrix.
Thereby condition $\Diamond$D holds and the convergence result in~\cite{Tsi84}  appears as a special 
 	case of our Theorem~\ref{thm:diamondD}.
  
Moreau~\cite{Mor05} and Hendrickx and Blondel~\cite{HB05} independently
	proved that in the synchronous case ($\Delta = 1$), asymptotic consensus is still guaranteed 
	when replacing the condition of bounded intercommunication intervals by a
	symmetry condition on the edges 
	of the communication graphs, namely $(i,j) \in E(t)$ iff $(j,i)\in E(t)$ for any $t\in\IN$.
The latter condition corresponds to a particular case of the decentralized model,
	and thus the convergence result in~\cite{Mor05,HB05}, as well as its extension 
	in~\cite{BHOT05} to the case of bounded 
	communication delays (assumption B), are contained in Theorem~\ref{thm:decent}.

Cao et al.~\cite{CSM05} proved that in the case of stochastic matrices with equal positive 
	entries in each row (usually referred to as the {\em equal neighbor model}\,), 
	and under assumptions A and B with $\Delta = 1$, the agreement 
	algorithm achieves asymptotic consensus when all the communication graphs are oriented.
Their convergence result thus coincides with Theorem~\ref{thm:coord}
	in the particular case of a synchronous multiagent system and with the equal neighbor model.

After writing the proof of Theorem~\ref{thm:coord}, we became aware of
	two recent papers both containing Theorem~\ref{thm:decent} in the synchronous case.
In~\cite{HT11}, Hendrickx and Tsitsiklis showed that the agreement algorithm achieves asymptotic consensus
 	under assumptions A, B with $\Delta =1$,  D1
 	and the so-called {\em cut-balance} condition.
In light of Proposition~\ref{pro:Pj} below, the latter condition
	turns out to correspond to the decentralized model.
Independently, Touri and Nedi\'c~\cite{TN11} established a general convergence result for the infinite
	product of random stochastic matrices, and to do that, they first proved that this 
	result holds in the deterministic case; see Lemma~5 in~\cite{TN11}.
It is easy to see that this lemma actually coincides with our Theorem~\ref{thm:decent}
	in the particular case of a synchronous multiagent system.
In fact, our technique for proving the lemmas in Sections~\ref{sec:column} and~\ref{sec:decent} {\it infra}
 	specialized to $\Delta = 1$, is similar 
	to the one used for the proof of the
	deterministic result in~\cite{TN11}.

\section{A seminorm for multiagent dynamics}\label{sec:seminorm}  

In this section, we discuss some auxiliary convergence results that will enter enter in
	our proofs in Section~\ref{sec:proofs}.
To illustrate their usefulness for the study of convergence of product of stochastic matrices,
	we also present a short proof of Wolfowitz's classical theorem based on them.
	
\subsection{An operator seminorm}

For any integer $N\geq 2$, we consider the seminorm on $\IR^N$ defined as the difference between the maximum and
	minimum entry of vector $x$
	$$\lVert x\lVert_{\bot} = \max_{i} (x_i) - \min_{i} (x_i) \enspace.$$
Thus $\lVert x\lVert_{\bot}$ is null if and only if $x\in \IR  \one$, where $\one$
	is the vector whose components are all equal to 1.
For any square matrix~$A$ with $\one$ as an eigenvector, the induced matrix seminorm $\lVert A\lVert_{\bot}$ 
	is $$ \lVert A \lVert_{\bot} = \sup _{x\notin \IR \one} 
		\frac{\lVert A x \lVert_{\bot}}{\lVert x\lVert_{\bot}} \enspace.$$
One key property of the seminorm $\lVert \cdot \lVert_{\bot}$ is that
			it is sub-multiplicative, i.e., 
			$$ \lVert AB \lVert_{\bot} \leq \lVert A \lVert_{\bot} \lVert B \lVert_{\bot} \enspace.$$
Another one is about the vectors that {\em realize} $\lVert A \lVert_{\bot}$.

\begin{lemma}\label{lem:realseminorm}
Let~$A$ be a square matrix with $\one$ as an eigenvector, 
	and let $\{e_i \, \mid\, i\in [N]\}$ denote the standard basis of $\IR^N$. 
There exists a nonempty subset $I$ of $[N]$ such that the vector $e_I = \sum_{i\in I} e_i$
	realizes $\lVert A \lVert_{\bot}$, i.e., 
	$$\lVert A \lVert_{\bot} = \lVert A \big ( e_I \big ) \lVert_{\bot} \enspace.$$
\end{lemma}	

\begin{proof}
By linearity, we have
	$$ \lVert A \lVert_{\bot} = \sup _{\lVert x\lVert_{\bot} = 1} 
					\lVert A(x)\lVert_{\bot} \enspace.$$
Let $x$ be a vector such that $\lVert x\lVert_{\bot} = 1$, and let $\overline{x}= x- x_{i_0} \one$ where
	$x_{i_0} = \min_i(x_i)$.
Then $\lVert \overline{x}\lVert_{\bot} = 1$, and $\overline{x}_{i_0} = \min_i (\overline{x}_i) = 0$.
Since $A(\one) = \one$,
 	$$\lVert A \lVert_{\bot} = \sup \left \{ \lVert A x \lVert_{\bot}  \, \mid \, 
	x\in\IR^N  \mbox{ with } \max_i (x_i)=1 \mbox{ and } \min_i (x_i)= 0 \right \} \enspace.$$
Moreover by compactness, there exists a vector $x$ with $\max_i (x_i)=1$ and $\min_i (x_i)= 0$,
	and such that $\lVert A \lVert_{\bot} = \lVert A x \lVert_{\bot}$.
	
Now suppose that some entries of $x$ are not in $\{0,1\}$, and let $x_j\in]0,1[$ be one of them.
We consider the two vectors $x^- = x- x_j \cdot e_j$ and $x^+ = x- (1-x_j) \cdot  e_j$, 
	and we denote $y = A x $, $y^-= A x^-$, and $y^+ = A x^+$.
Then for any index $i$, we have $y^-_i = y_i - x_j A_{i j}$,
	$y^+_i = y_i + (1- x_j) A_{i j}$, and so 
	$y^-_i \leq y_i \leq y^+_i$.
Since $ \lVert x \lVert_{\bot} = \lVert x^-\lVert_{\bot} = \lVert x^+\lVert_{\bot} = 1$, 
	$\lVert y^- \lVert_{\bot}$
	and $\lVert y^+\lVert_{\bot} $ are both less or equal to $\lVert A \lVert_{\bot}$.
Let $i_0$ and $i_1$ be two indices such that $\lVert y\lVert_{\bot} = y_{i_1} - y_{i_0}$.
Then $$ y^-_{i_1} - y^-_{i_0} = \lVert A \lVert_{\bot} - x_j (A_{i_1 j} - A_{i_0 j}) \mbox{ and }
	y^+_{i_1} - y^+_{i_0} = \lVert A \lVert_{\bot} + (1- x_j) (A_{i_1 j} - A_{i_0 j}) \enspace.$$
From $x_j \in ]0,1[$, $y^-_{i_1} - y^-_{i_0} \leq \lVert A \lVert_{\bot}$, and  
	$y^+_{i_1} - y^+_{i_0} \leq \lVert A \lVert_{\bot}$, we derive that 
	$A_{i_1 j} = A_{i_0 j}$, and so $\lVert y^-\lVert_{\bot} = \lVert y^+\lVert_{\bot} = \lVert A \lVert_{\bot}$. 
One by one, we thus eliminate all the entries of $x$ different from 0 and 1, 
	and obtain a vector of the desired form.
\end{proof}
Observe that if $e_I = \sum_{i\in I} e_i$ realizes $\lVert A \lVert_{\bot}$, then 
	so does $e_{\overline{I}}$, where $\overline{I}$ denotes the complement of
	$I$ within $[N]$.

The vector $\one$ is an eigenvector of each stochastic matrix, and we easily check	
	that for any stochastic matrix~$A$ and any vector $x \in \IR^N$,
	$$ \lVert A x\lVert_{\bot} \leq \lVert x\lVert_{\bot} \enspace.$$
It follows that  the induced matrix seminorm of a stochastic matrix 
	 is less or equal to 1.
	
Interestingly we can compare $\lVert A \lVert_{\bot}$ with the {\em coefficients of ergodicity} of $A$ 
	previously introduced in the literature~\cite{Haj58, Wol63}, namely
	$$\delta(A) = \max_{j} \max_{i_1,i_2} |A_{i_2 j} -A_{i_1 j}|\enspace,$$
	and
	$$\lambda(A) = 1- \min_{i_1,i_2} \sum^N_{j=1} \min (A_{i_1 j}, A_{i_2 j}) \enspace.$$

\begin{proposition}\label{prop:deltalambda}
Let $A$ be a stochastic matrix.
Then
	$$ \delta(A) \leq  \lVert A \lVert_{\bot} \leq \lambda(A)\enspace.$$
Moreover $\lambda(A) =1$ if and only if $\lVert A \lVert_{\bot} =1$.
\end{proposition}

\begin{proof}
First we observe that $$ \delta(A) = \max_{j=1,\cdots,N} \lVert A e_j \lVert_{\bot}\enspace,$$
	and the inequality $ \delta(A) \leq  \lVert A \lVert_{\bot}$ immediately follows.

For the second inequality, we consider a realizer of $\lVert A \lVert_{\bot}$ that we
	denote $e_I = \sum_{i \in I} e_i$ (cf. Lemma~\ref{lem:realseminorm}).
Let $f= A e_I$, $f_{\alpha} = \max_i f_i$, and $f_{\beta} = \min_i f_i$.
Then we have $$ \lVert A \lVert_{\bot} = \sum_{j\in I} (A_{\alpha j} - A_{\beta j})	\enspace.$$
Since $A$ is a stochastic matrix, we get
	$$1-\lVert A \lVert_{\bot} = \sum_{j=1}^N  A_{\alpha j} - \sum_{j\in I} (A_{\alpha j} - A_{\beta j})	
	                           = \sum_{j\notin I} A_{\alpha j} +  \sum_{j\in I} A_{\beta j}\enspace.$$
Hence $$1-\lVert A \lVert_{\bot} \geq \sum_{j=1}^N \min (A_{\alpha j} , A_{\beta j})\enspace,$$ and the 
	inequality $\lVert A \lVert_{\bot} \leq \lambda(A)$ immediately follows.
	
Suppose now that $\lambda(A) =1$, i.e., there exist two indices $i_1, i_2$ such that 
	for each $j\in [N]$, either $A_{i_1j}=0$ or $ A_{i_2j}=0$.
Let $I=\{j\in [N] \mid  A_{i_1j} \neq 0\}$, $e_I = \sum_{i \in I} e_i$,
 	and $f= A e_I$.
Since $A$ is stochastic,  its $i_1$-th and $i_2$-th rows each contains a non-null entry,
 	and thus neither $I$ or its complement  is empty.
Hence $$ f_{i_1} = 1 \mbox{ and } f_{i_2} = 0 \enspace,$$
	which shows  $\lVert A \lVert_{\bot} =1$.
\end{proof}

We now give a corollary that is useful for convergence proofs.

\begin{corollary}\label{cor:boundnorm}
Let $A$ be a $ N \times N$ stochastic matrix, $j$ be any index in $[N]$,
	and let $\beta_j$ be the minimum of all 
	the entries in the  $j$-th column, i.e.,
	$ \beta_j = \min \{ A_{i,  j}  \mid   i \in  [N]  \}$.
Then, $$\lVert A \lVert_{\bot} \leq 1- \sum_{j=1}^N  \beta_j \enspace.$$	
In particular, if all the entries of one column of~$A$ are positive,
	then $\lVert A \lVert_{\bot} < 1$.
\end{corollary}

\begin{proof}
Since all $A$'s entries are non-negative, we have
	$$\lambda(A) \leq  1- \sum_{j=1}^N  \beta_j \enspace.$$
Using Proposition~\ref{prop:deltalambda}, we derive that $\lVert A \lVert_{\bot}\, \leq 1- \beta_j$.
In the case $\beta_j >0$ for some index~$j$, we obtain $\lVert A \lVert_{\bot} < 1$.	
	\end{proof}

\subsection{A simple criterion for convergence}\label{subsec:criterion}

Now we give a seminorm based condition on a sequence of stochastic matrices under 
	which their product converges to a rank one stochastic matrix. 
This criterion lies implicitly behind several convergence proofs 
	(for instance, see~\cite[Section~7.3]{BT89} or \cite{BHOT05,TN11} and also
	\cite{GQ13} for related results concerning Markov operators on cones).

\begin{proposition} \label{pro:convergence}
For each integer $t\in \IN$, let $A(t)$ be a stochastic matrix, and 
	let $P(t) = A(t)\dots A(0)$.
The following conditions are equivalent:
\begin{enumerate}
	\item The sequence $\big ( P(t)\big )_{t\in \IN}$ converges  to a  matrix of the form 
	$\one \pi^T$ where $\pi$ is a probability vector in $\IR^N$.

	\item The sequence $\big (  \lVert P(t) \lVert_{\bot}\big )_{t\in \IN}$ converges  to 0.
	
	\item For each vector $v\in \IR^N$, the sequence $\big ( P(t)v\big )_{t\in \IN}$ converges  
	to some vector in the line $\IR \one$.

	\item For each vector $v\in \IR^N$, the sequence $\big (  \lVert P(t) v\lVert_{\bot}\big )_{t\in \IN}$ 
	converges to 0.
	
	\end{enumerate} 
\end{proposition}

\begin{proof}
The implications $(1) \Rightarrow (2)$,  $(3) \Rightarrow (4)$, $(1) \Rightarrow (3)$,
	and $(2) \Rightarrow (4)$ are obvious.
	
To show that $(4)\Rightarrow (3)$, we consider a sequence of vectors $x(t) = P(t)v$ 
	where $v$ is some vector in $\IR^N$. 
We denote 
	$$M(t)  =  \max_i \big( x_i(t)\big )  \mbox{ and } m(t) =  \min_i \big( x_i(t)\big )\enspace.$$
Hence $\lVert x(t) \lVert_{\bot} = M(t) - m(t) $, and  $(4)$ is equivalent to
	$\lim_{t \rightarrow \infty}  M(t) - m(t) = 0$.
Since each matrix $A(t)$ is stochastic, the sequences $\big(M(t)\big)_{t\in \IN}$ and $\big(m(t)\big)_{t\in \IN}$
	are non-increasing and non-decreasing, respectively.
It follows that the latter two sequences as well as all the sequences $\big( x_i(t)\big )_{t\in \IN}$ are 
	convergent to the same limit, which shows (3).

For the implication $(3) \Rightarrow (1)$, suppose that (3) holds.
In particular with $v=e_j$,  each sequence $\big( P_{ij}(t)\big )_{t\in \IN}$ converges to 
	some scalar $
	j$ independent of index~$i$.
Therefore, the sequence $\big ( P(t)\big )_{t\in \IN}$ converges, and 
	$\lim_{t\rightarrow +\infty} P(t) = \one \pi^T$.
Since each matrix $P(t)$ is stochastic and the set of stochastic matrices is 
	closed, $\pi$ is a probability vector.
\end{proof}

As an immediate consequence of Corollary~\ref{cor:boundnorm}, the sub-multiplicativity
 	of the seminorm $\lVert \cdot \lVert_{\bot}$, and Proposition~\ref{pro:convergence},
	we obtain the well-known result that for any ergodic stochastic matrix $A$,  
	$\lim_{t \rightarrow \infty} A^t$ exists and is a rank one stochastic matrix.
In fact, we can even derive Wolfowitz's theorem~\cite{Wol63} which generalizes 
	the latter result to infinite products of matrices taken from a finite set 
	of ergodic stochastic matrices.

\begin{theorem}[Wolfowitz]
Let ${\cal M}$ be a nonempty finite set of stochastic matrices such that any finite
	product of matrices in this set is ergodic.
For each $t\in \IN$, let $A(t)$ be a matrix in ${\cal M}$. 
Then $\lim_{t\rightarrow +\infty} A(t) \cdots A(0)$ exists, and the limit
	 is of the form $\one \pi^T$ where $\pi$ is a probability vector in $\IR^N$.
\end{theorem}

\begin{proof}
We begin by mimicking the first steps of the proof in~\cite{Wol63} with the seminorm 
	$\lVert \cdot \lVert_{\bot}$ to be substituted for the coefficient of
 	ergodicity $\lambda$.
Then instead of Theorem~2 of~\cite{Haj58} used by Wolowitz, we just need the 
	sub-multiplicativity of the seminorm to conclude.
	
In more detail, we first show the following lemma.
\begin{lemma}\label{lem:wol}
The seminorm of any product of $N^2 + 1$ matrices in ${\cal M}$ is less than 1.	
\end{lemma} 
\begin{proof}
Let $\sim$ be the equivalence relation defined on the set of square stochastic 
	matrices by
	$A\sim B$ iff $A$ and $B$ have the same communication graph.
We easily check that $\sim$ is preserved by (right) multiplication
	with stochastic matrices.
Moreover,  the conditions $\lambda(A) = 1$ and $\lambda(B) = 1$ are clearly
	equivalent when $A\sim B$. 
Using Proposition~\ref{prop:deltalambda} twice, we derive that
	if $A\sim B$ and $\lVert A \lVert_{\bot} = 1$, then  $\lVert B \lVert_{\bot} =1$.

Let $A_0,\cdots, A_{N^2}$ be $N^2 +1$ matrices in ${\cal M}$.
Since there are at most $N^2$ equivalence classes under the relation~$\sim$, there exist 
	two indices $k,\ell$, $0\leq k < \ell\leq N^2$, such that
	$$ A_{N^2}\cdots A_{\ell} \sim A_{N^2}\cdots A_{k} \enspace.$$
Let $A= A_{N^2}\cdots A_{\ell}$ and $B= A_{\ell-1} \cdots A_k$;
	we have $AB \sim A$.
It follows that for any positive integer $n$, $AB^n \sim A$.
Moreover by assumption on ${\cal M}$, the matrix $B$ is ergodic, i.e., $B^{n_0} >0$
 	for some positive integer $n_0$.
	
Now suppose that $\lVert A \lVert_{\bot} = 1$.
The above argument shows that  $\lVert AB^{n_0} \lVert_{\bot} = 1$ 
	and by the sub-multiplicativity of~$\lVert \cdot \lVert_{\bot}$, 
	we get $\lVert B^{n_0} \lVert_{\bot} = 1$, a contradiction with 
	Corollary~\ref{cor:boundnorm}.
Therefore, $\lVert A \lVert_{\bot} < 1$.
Using the sub-multiplicativity of~$\lVert \cdot \lVert_{\bot}$ again, we obtain
	$\lVert A_{N^2}\cdots A_{0} \lVert_{\bot} < 1$ as required.
\end{proof}
Let $\delta$ denote the supremum of the seminorms of the matrices that  are
	a product of $N^2 + 1$ matrices in ${\cal M}$.
Since ${\cal M}$ is finite, there is a finite number of such matrices and
	Lemma~\ref{lem:wol} shows that $\delta <1$.
Using the sub-multiplicativity of the seminorm, the theorem immediately follows. 
\end{proof}

In the next section, we show how to use the criterion 
	in Proposition~\ref{pro:convergence}
	to prove convergence theorems where the finiteness assumption of the set ${\cal M}$
	is weakened by assuming a uniform positive lower bound on positive entries
	(assumption A3),
	and where the ergodicity assumption is replaced by conditions 
	that basically guarantee ``eventual positivity'' of some finite products of matrices
	in~${\cal M}$.

\section{Convergence with bounded delays and topological changes}\label{sec:proofs} 

This section is devoted to the proof of Theorems~\ref{thm:coord} and~\ref{thm:decent}.
In Section~\ref{sec:red}, we introduce a family of stochastic matrices $(A^{\Delta}(t))_{t\in\IN}$
	of size $\Delta N$, which allows us to express our original evolution equation~(\ref{eq:linear})
	as a ``zero delay'' evolution equation of the form $$X(t+1) = A^{\Delta}(t)X(t)$$
	on suitably defined vectors $X(t)$.
In Section~\ref{sec:column}, we investigate the time evolution
 	of the sets of positive entries in each column of the successive products of the matrices~$A^{\Delta}(t)$.
The content of this subsection constitutes the core of our proofs of Theorems~\ref{thm:coord} and~\ref{thm:decent}.
Combined with a simple combinatorial argument, it allows us to conclude
	in Section~\ref{sec:coord} for the case of the coordinated model. 
Section~\ref{sec:decent} is dedicated to the study of the decentralized model.
We start by deriving from condition D1 the eventual positivity of some columns 
	of the successive products of the matrices~$A^{\Delta}(t)$.
Then we refine the ``stationarity'' condition elaborated in Section~\ref{sec:column}
	to prove that asymptotic consensus in the decentralized model.

\subsection{Reduction to the zero delay case}\label{sec:red}

We mimic the classical reduction of a $\Delta$-th order ordinary differential
	equation to a system of $\Delta$ ordinary differential equations of first order.
For any time $t\geq \Delta -1$, let $X(t)_{(\delta,i)}$ denote the family
	in $\IR^{[\Delta] \times [N]}$ defined by
	$$ X(t)_{(\delta,i)} = x_i(t-\delta +1 )\enspace,$$
		where $\delta \in [\Delta]$ and $i\in [N]$.
Letting  $\delta_{ij}(t) = t - \tau_{ij}(t) +1 \in \{0,\cdots, \Delta -1\}$, 
	equation~(\ref{eq:linear}) can be rewritten in the following way 
	$$ X(t + 1)_{(\delta,i)} =  \left\{ \begin{array}{ll}
	        X(t)_{(\delta - 1,i)} & \mbox{ if } \delta \in \{2,\cdots,\Delta\}     \\ \\
	        \sum_{j=1}^N A_{i,j}(t) X(t)_{(\delta_{ij}(t),\,j)} & \mbox{ if } \delta = 1 \enspace,
	\end{array}\right. $$ 
	which  is equivalent to 
$$  X(t + 1)_{(\delta,i)} =  
	\sum_{(\delta',\,j) \in [\Delta] \times [N]} A^{\Delta}_{(\delta,i),(\delta',\,j)}(t) X(t)_{(\delta',\,j)} $$
	with $$ A^{\Delta}_{(\delta,i),(\delta',\,j)}(t) =  \left\{ \begin{array}{ll}
	            1 & \mbox{ if } i =j,\ \delta' = \delta - 1, \mbox{ and } \delta \in \{2,\cdots,\Delta\} \\ 
	            A_{i,j}(t) & \mbox{ if } \delta = 1  \mbox{ and } \delta' = \delta_{ij}(t) \\
	            0 & \mbox{ otherwise.}
     	\end{array}\right. $$
The key point here is that  from time $\Delta -1$ on, the vector $X$ is updated according to the linear
			equation  with ``zero delay'' $$X(t+1) = A^{\Delta}(t)X(t)\enspace.$$
			
For ease of notation, in what follows we encode each ordered pair 
	$(\delta,i)$ in $[\Delta] \times [N]$  into
	the integer $k = \Delta i - \delta +1$ in $[\Delta N]$.
Then the vector $X(t)$ is  in $\IR^{\Delta N}$ and  $A^{\Delta}(t)$ is a $\Delta N \times \Delta N$ matrix.
The updating rule for $X(t)$ can be rewritten
	$$  X_k(t+1) =  \left\{ \begin{array}{ll}
	         X_{k+1}(t)  & \mbox{ if $k$ is not a multiple of } \Delta    \\ \\
	        \sum_{j=1}^N A_{i,j}(t) X_{\Delta j - \delta_{ij}(t) + 1}(t) & \mbox{ otherwise.}
	\end{array}\right. $$
Using the Kronecker delta $\delta_{m,n}$ (not to be confused with the delays $\delta_{ij}(t)$),
	the above expression of~$A^{\Delta}_{(\delta',j),(\delta,i)}(t)$ implies that
\begin{enumerate}
	\item $A^{\Delta}_{m,n}(t) = \delta_{m + 1,n}$ if $m$ is not a multiple of $\Delta$;
	\item all the entries $A^{\Delta}_{\Delta i, \Delta (j-1) +1}(t), \dots, A^{\Delta}_{\Delta i, \Delta j}(t)$ 
		are null except one which is equal to $A_{i,j}(t)$;
	\item $A^{\Delta}_{\Delta i, \Delta i}(t) = A_{i,i}(t)$.
\end{enumerate}
We observe that each matrix $A^{\Delta}(t)$ is stochastic (A1), and every positive
	entry of $A^{\Delta}(t)$ is at least equal to $\alpha$ (A3).
However neither A2 nor D1 holds for the matrices $A^{\Delta}(t)$, and the above
	reduction does not allow us to limit ourselves
	to the zero delay case $\Delta =1$ while maintaining the basic assumptions A and B.

We now choose some time $t_0\geq \Delta -1$ which stays fixed in the rest of the section 
	except in the very last steps of the proofs of Theorems~\ref{thm:coord} and~\ref{thm:decent}.
For any time $t \geq t_0$, we let $$P(t)= A^{\Delta}(t)\dots A^{\Delta}(t_0) \enspace.$$
Because of the above mentioned properties of the matrix $A^{\Delta}(t)$, we have the following
	recurrence relations for $P(t)$'s entries:

\begin{description}
	\item[R1:] $P_{m, \Delta j}(t +1) =  P_{m+1, \Delta j}(t)$, for each index $m$ that is not a multiple of $\Delta$;
	\item[R2:]  $P_{\Delta i, \Delta j}(t +1) = \sum_{k=1}^N A_{i, k}(t+1) P_{m_k, \Delta j}(t)$, where
	for each index $k$ in the sum, $m_k$ is some integer in $\{\Delta (k-1) +1, \dots, \Delta k \}$, 
	and where $m_k = \Delta k $ for $k=i$.
\end{description}
	
\subsection{Positive entries in the $\Delta j$-th column}\label{sec:column}

In this section, we fix some index $j\in [N]$, and study the entries
	in the $\Delta j$-th column of~$P(t)$.
We define the two sets
	$$ S_j^{\Delta}(t)  =  \{ m \in [\Delta N] \mid P_{m, \Delta j}(t)>0 \}
	 \ \mbox{ and }\ S_j(t)  =  \{ i \in [N] \mid P_{\Delta i, \Delta j}(t)>0 \} 
	\enspace.$$

\begin{lemma}\label{lem:S0}
$\Delta j \in S_j^{\Delta}(t_0)$.	
\end{lemma}
\begin{proof}
By definition, $P(t_0) = A^{\Delta}(t_0)$.
The lemma directly follows from property (3) of the matrix~$A^{\Delta}(t_0)$, and from 
	assumption A2 on the matrix~$A(t_0)$.
\end{proof}

\begin{lemma}\label{lem:Sincrease}
For all $t\geq t_0 ,\ S_j(t) \subseteq S_j(t+1)$.	
\end{lemma}
\begin{proof}
From the recurrence relation R2, we deduce that 
	$$P_{\Delta i, \Delta j}(t +1) \geq A_{i, i}(t+1) \, P_{\Delta i, \Delta j}(t)\enspace,$$
	and the lemma follows from assumption A2 on the matrix~$A(t+1)$.
\end{proof}

\begin{lemma}\label{lem:SDeltaincrease}
For all $t\geq t_0 ,\ S_j^{\Delta}(t) \subseteq S_j^{\Delta}(t+1)$.	
\end{lemma}
\begin{proof}
	First we give a more precise description of $S_j^{\Delta}(t_0)$.
	Since $P(t_0) = A^{\Delta}(t_0)$,  properties~(1) and~(2) of the matrix $A^{\Delta}(t_0)$ implies that
		$$ m \in S_j^{\Delta}(t_0) \Leftrightarrow \left \{ \begin{array}{l}
		                                         m = \Delta i \ \mbox{ and } 
		\  A^{\Delta}_{\Delta i, \Delta j}(t_0) >0 \\
		                  \hspace{1cm}                       \mbox{ or } \\
		                                         m = \Delta j - 1 \enspace.
	                                            \end{array} \right. $$
	
	Let $t\geq t_0$, and $m \in S_j^{\Delta}(t)$.
	We denote $t = t_0 + \theta$, and $m = \Delta i - \ell$ with 
	$i\in [N]$ and $\ell\in \{0, \cdots, \Delta - 1\}$.
	We consider two cases:
	\begin{enumerate}
		\item $\theta \geq \ell$. \\
		By as many applications as possible  of the recurrence relation~R1, we obtain
			$$P_{m,\Delta j}(t) = P_{\Delta i,\Delta j}(t_0 +\theta - \ell) \ \mbox{ and } \ 
			  P_{m,\Delta j}(t + 1) = P_{\Delta i,\Delta j}(t_0 + 1 + \theta  - \ell)  \enspace.$$
		The lemma follows from Lemma~\ref{lem:Sincrease} in this case.
		\item $\theta < \ell$. \\
		Then $\ell \geq 1$, i.e., $m$ is not a multiple of $\Delta$,
		 	and we can apply the recurrence relation~R1 to obtain 
		$$P_{m,\Delta j}(t) = P_{\Delta i - \ell +\theta ,\Delta j}(t_0) \ \mbox{ and } \ 
		  P_{m,\Delta j}(t + 1) = P_{\Delta i - \ell +\theta - 1 ,\Delta j}(t_0)  \enspace.$$	
	\end{enumerate}
		In this case, $\Delta i - \ell +\theta$ is not a multiple of $\Delta$,
		and from the above description of $S_j^{\Delta}(t_0)$ follows that 
		$\Delta i - \ell +\theta = \Delta j - 1$.
		Hence $P_{m,\Delta j}(t + 1) =  A_{j,\,j}(t_0)$.
		The lemma follows from assumption A2 on the matrix~$A(t_0)$ in this case.
\end{proof}

By Lemmas~\ref{lem:S0} and~\ref{lem:Sincrease}, each set $S_j^{\Delta}(t)$ is non-empty.
We define $\pi_j(t)$ as the minimum positive entry in the $\Delta j$-th column
	of the matrix $P(t)$, or
	$$\pi_j(t) = \min\{ P_{m, \Delta j}(t)  \mid  m \in  S_j^{\Delta}(t) \} \enspace .$$
	
\begin{lemma}\label{lem:phi}
For all $t\geq t_0,\ \pi_j(t +1)\geq \alpha \, \pi_j(t)$.	
\end{lemma}
\begin{proof}
Let $m \in S_j^{\Delta}(t+1)$, i.e., $P_{m , \Delta j}(t +1) >0$.
There are two cases to consider:
\begin{enumerate}
	\item $m$ is  a multiple of $\Delta$, i.e., $m=\Delta i$ for some $i\in [N]$.
	Relation~R2 implies that 
		$$P_{\Delta i, \Delta j}(t +1) \geq A_{i, i}(t+1)\  P_{\Delta i, \Delta j}(t)  \enspace.$$
	Since the matrix $A(t+1)$ satisfies A2 and A3, we  have
	$P_{\Delta i, \Delta j}(t +1) \geq \alpha \, \pi_j(t)$.
	\item $m$ is not a multiple of $\Delta$.
	By relation~R1, we obtain $$P_{m , \Delta j}(t +1) = P_{m +1, \Delta j}(t) \enspace.$$
	Hence $m + 1 \in S_j^{\Delta}(t)$,
	and thus $P_{m , \Delta j}(t +1) \geq \pi_j(t)$.
	
\end{enumerate}
In both cases we  have $P_{m, \Delta j}(t +1) \geq \alpha \, \pi_j(t)$, as needed.
\end{proof}

\begin{lemma}\label{lem:key}
Let $t\geq t_0$.
If $S_j^{\Delta}(t) = S_j^{\Delta}(t+1)$, then for each index $i$ in $S_j(t)$,  all the entries
 	$P_{\Delta (i-1) + 1, \Delta j}(t), \dots , P_{\Delta i, \Delta j}(t)$  are positive.
\end{lemma}

\begin{proof}
Let $i \in S_j(t)$.
By decreasing induction on $m$, $\Delta\,(i-1) \leq m \leq \Delta\, i$, 
	we prove 	that $$ P_{m, \Delta j}(t)  > 0\enspace.$$
\begin{enumerate}
	\item The basic case $m = \Delta\, i$ corresponds to  $i \in S_j(t)$.
	 
	\item For the inductive step, let us assume that $ P_{m, \Delta j}(t)\!>\! 0$
	with $m\!\in\!\{\Delta\,(i\!-\!1)\!+\! 1, \cdots, \Delta\, i\}$.
	Therefore, $m-1$ is not a multiple of $\Delta$ and by relation~R1, we have 
	$$P_{m-1, \Delta j}(t+1) = P_{m, \Delta j}(t) \enspace.$$
	Hence $m-1\in S_j^{\Delta}(t+1)$.
	Since $S_j^{\Delta}(t) = S_j^{\Delta}(t+1)$, we obtain  $P_{m, \Delta j}(t) > 0$, as required.
\end{enumerate}
\end{proof}

\begin{lemma}\label{lem:incoming}
If $S_j^{\Delta}(t) = S_j^{\Delta}(t+1)$, then $S_j(t)$ has  no incoming edge in the graph $G(t+1)$.
\end{lemma}
\begin{proof}
By contradiction, suppose that $S_j^{\Delta}(t) = S_j^{\Delta}(t+1)$, and that 
	$S_j(t)$ has  an incoming edge in the graph~$G(t+1)$.
Hence  there exist $i$ and $k$ in $[N]$ such that 
	$ P_{\Delta i, \Delta j}(t) = 0$, $P_{\Delta k, \Delta j}(t) >0$, and $A_{i,k}(t+1) >0$.
Since $S_j^{\Delta}(t) = S_j^{\Delta}(t+1)$, we have $$ P_{\Delta i, \Delta j}(t+1) = 0 \enspace.$$
Moreover Lemma~\ref{lem:key} states that 
	$$\forall m \in \{\Delta\, (k-1) +1, \cdots, \Delta\, k\} : \  P_{m, \Delta j}(t) > 0 \enspace,$$ 
	and relation~R2 implies that 
	$$P_{\Delta i, \Delta j}(t+1) \geq A_{i,k}(t+1) \ P_{\ell, \Delta j}(t)$$ for 
	some $\ell \in \{\Delta\, (k-1) +1, \cdots, \Delta\, k\}$.
Hence $P_{\Delta i, \Delta j}(t+1) >0$, a contradiction.
\end{proof}

\subsection{Convergence in the coordinated model}\label{sec:coord}

In this section, we consider the coordinated model (assumptions A, B, and C),
	and prove our first convergence theorem.

We start by specializing Lemma~\ref{lem:incoming} to the 
	case of oriented communication graphs.
Let us define
	$$ S^{\Delta}(t)  =  \{ (m,j) \in [\Delta N]\times [N] \mid m \in S_j^{\Delta }(t) \}
		\enspace.$$	

\begin{lemma}\label{lem:SDeltacoord}	
If  $G(t+1)$ is $j$-oriented, then either $S^{\Delta}(t)\neq S^{\Delta}(t+1)$,
	or $S_j^{\Delta}(t) =[\Delta N]$. 
\end{lemma}

\begin{proof}
Suppose that $S^{\Delta}(t) = S^{\Delta}(t+1)$.
Therefore $S_j^{\Delta}(t) = S_j^{\Delta}(t+1)$.
Lemma~\ref{lem:incoming} ensures that $S_j(t)$ has no incoming link in~$G(t+1)$.
Since $G(t+1)$ is $j$-oriented, $S_j(t) = [N]$.
By  Lemma~\ref{lem:key}, we conclude that $S_j^{\Delta}(t) = [\Delta N]$, as required.
\end{proof}

The latter lemma allows us to show that  there is an agent~$j$ such that  
	$S_j^{\Delta}$ is 
	equal to  $[\Delta N]$ by time $t_0 + \Delta N^2 -2N+1$.

\begin{proposition}\label{pro:ergodicfini}
In the coodinated model, there exists an agent~$j$ such that 
	$$ \forall m\in [\Delta N] \ :\ P_{m,\Delta j} (t_0 + \Delta N^2 -2N+1) > 0 \enspace. $$
\end{proposition}	
\begin{proof}
Let $t\geq t_0$.
As immediate consequences of Lemma~\ref{lem:S0} and Lemma~\ref{lem:SDeltaincrease},
 	respectively, we have $$|S^{\Delta}(t_0) | \geq N $$ 
			and  for all $t\geq t_0$,
			$$ S^{\Delta}(t) \subseteq S^{\Delta}(t+1) \enspace.$$
Using Lemmas~\ref{lem:SDeltacoord} and~\ref{lem:SDeltaincrease}, we obtain that 
	either $|S^{\Delta}(t) | \geq N +t - t_0 $ or $S_j^{\Delta}(t) = [\Delta N]$
	for some~$j$ in $[N]$.
We observe that if 	the cardinality of $S^{\Delta}(t)$ is greater than $\Delta N^2 - N$,
	then the matrix $P(t)$ has at least one of its
		$\Delta j$-th columns  with positive entries to complete the proof of the lemma.
	\end{proof}

\vspace{0.2cm}		
We are now in position to prove Theorem~\ref{thm:coord}.
Let us consider an agent~$j$ such that 
	$$ S_j^{\Delta}(t_0 + \Delta N^2 -2N+1) = [\Delta N] $$
	the existence of which is ensured by Proposition~\ref{pro:ergodicfini}.
By Lemma~\ref{lem:phi}, we have
		$$\pi_j(t_0 + \Delta N^2 -2N+1) \geq \alpha^{\Delta N^2 -2N+1} \enspace.$$
By Corollary~\ref{cor:boundnorm}, we derive 
	$ \lVert P(t_0 + \Delta N^2 -2N+1) \lVert_{\bot} \leq 1 - \alpha^{\Delta N^2 -2N+1}$.
In other words, we have shown that for any $t_0 \geq \Delta -1$
		$$\lVert A^{\Delta}(t_0 + \Delta N^2 -2N+1)\cdots  A^{\Delta}(t_0) \lVert_{\bot} 
		\leq 1 - \alpha^{ \Delta N^2 -2N+1}\enspace.$$ 
Together with the sub-multiplicativity of the seminorm $\lVert  \cdot \lVert_{\bot}$,
	this implies that 
	$$\lim_{t\rightarrow +\infty} \lVert A^{\Delta}(t)\cdots  A^{\Delta}(0) \lVert_{\bot} =0 \enspace.$$
Theorem~\ref{thm:coord} then follows from Proposition~\ref{pro:convergence} and the definition of
		 vector $X(t)$.

\subsection{Convergence in the decentralized model}\label{sec:decent}

We now consider the decentralized model (assumptions A, B, and D).
Under the sole assumptions A and B, the set $ S^{\Delta}$ may remain small
	forever: for example, in the case of the sequence of the powers of the unit matrix, $ S^{\Delta}$ is 
	constantly equal to the diagonal in $[N]^2$.
Firstly, we show that assumption~D1 ensures that $ S^{\Delta}$ is eventually equal to $[\Delta N] \times N$.
However, D1 provides no bound on the time required the set $ S^{\Delta}$ to be maximal,
	and so under assumptions A, B, and D1,   no positive lower bound on the positive 
	entries of the matrix $P$ are guaranteed.
In the second part of this section, we show that assumption D2  allows us to control 
	how functions $\pi_j$ can decrease in time.
	
\vspace{0.2cm}
We start by refining the strong connectivity property of the graph
	$G_0 = ([N]\,, \cup_{t\geq t_0} E(t))$ ensured by~D1.
	
\begin{lemma}\label{lem:permanentconnectivity}	
For each $i\in [N]$, there exists a path  $i_n, \cdots, i_0$ in the graph $G_0$
	starting at $i_n=i$, ending at $i_0 =j$, and such that each edge $(i_k,i_{k-1})$
	is in $E(t_k)$ with 
	$t_k \geq t_{k-1} + \Delta$.
\end{lemma}

\begin{proof}
We consider the graph $G^{\infty} = ([N]\,, E^{\infty})$ where $E^{\infty}$
	is the set of edges in $G_0$ that occur infinitely often.
The lemma immediately follows from~D1.
\end{proof}

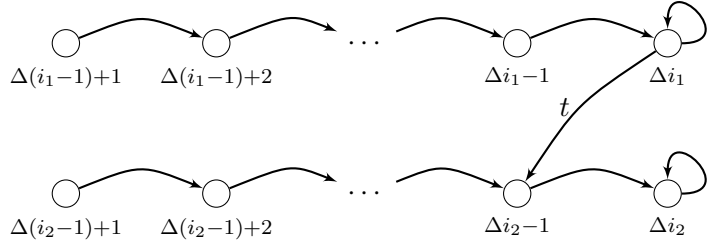
\begin{figure}\label{fig:tn}
\centering
\begin{tikzpicture}[>=latex',scale=1.0]
	\node[label=below:$\scriptstyle\Delta(i_1-1)+1$,shape=circle,draw] (a1) at (-6,0) {};
	\node[label=below:$\scriptstyle\Delta(i_1-1)+2$,shape=circle,draw] (b1) at (-4,0) {};
	\node (c1) at (-2,0) {$\cdots$};
	\node[label=below:$\scriptstyle\Delta i_1-1$,shape=circle,draw] (d1) at (0,0) {};
	\node[label=below:$\scriptstyle\Delta i_1$,shape=circle,draw] (e1) at (2,0) {};
	\draw[thick,->] (a1) .. controls +(1,0.4) ..  (b1);
	\draw[thick,->] (b1) .. controls +(1,0.4) ..  (c1);
	\draw[thick,->] (c1) .. controls +(1,0.4) ..  (d1);
	\draw[thick,->] (d1) .. controls +(1,0.4) ..  (e1);
	\draw[thick,->] (e1) .. controls +(1,0.0) and +(0.0,1) ..  (e1);

	\node[label=below:$\scriptstyle\Delta(i_2-1)+1$,shape=circle,draw] (a2) at (-6,-2) {};
	\node[label=below:$\scriptstyle\Delta(i_2-1)+2$,shape=circle,draw] (b2) at (-4,-2) {};
	\node (c2) at (-2,-2) {$\cdots$};
	\node[label=below:$\scriptstyle\Delta i_2-1$,shape=circle,draw] (d2) at (0,-2) {};
	\node[label=below:$\scriptstyle\Delta i_2$,shape=circle,draw] (e2) at (2,-2) {};
	\draw[thick,->] (a2) .. controls +(1,0.4) ..  (b2);
	\draw[thick,->] (b2) .. controls +(1,0.4) ..  (c2);
	\draw[thick,->] (c2) .. controls +(1,0.4) ..  (d2);
	\draw[thick,->] (d2) .. controls +(1,0.4) ..  (e2);
	\draw[thick,->] (e2) .. controls +(1,0.0) and +(0.0,1) ..  (e2);

	\draw[thick,->] (e1) .. controls +(-1.2,-0.8)  .. node[left] {$t$} (d2);
\end{tikzpicture}
\caption{Structure of $G^{\Delta}(t)$}
\end{figure}

\begin{proposition}\label{pro:ergodic}
In the decentralized model, for every $j$ in $[N]$, there exists some time $\theta_j \geq t_0$ such that for  
	each 
	$$ \forall m\in [\Delta N]\ : \ P_{m,\Delta j}(\theta_j)  >0 \enspace.$$
\end{proposition}	
\begin{proof}
At each time $t\in \IN$, let us consider the directed graph $G^{\Delta}(t)$
	to be the communication graph of~$A^{\Delta}(t)$, i.e., 
	$$G^{\Delta}(t) = ([\Delta N]\,, E^{\Delta}(t)) \enspace, $$ where	
	$(m,n)\in E^{\Delta}(t)$ if and only if $A^{\Delta}_{m,n}(t) >0$, 
	and let $$G^{\Delta}_0 = ([\Delta N]\,, \cup_{t\geq t_0} E^{\Delta}(t)) \enspace.$$
Properties~(1) and~(3) of the matrix $A^{\Delta}(t)$ imply that for each index $i\in [N]$, 
	the sequence of nodes $\Delta (i-1)+1, \cdots, \Delta i$ is a path in~$G^{\Delta}(t)$,
	and that there is a self-loop at node $\Delta i$.
Further if $A_{i_1,i_2}(t) >0$, then property~(2) for $A^{\Delta}(t)$ guarantees that 
	there is an edge in $G^{\Delta}(t)$ from node $\Delta i_1$ to some node
	in the path $\Delta (i_2-1)+1, \cdots, \Delta i_2$ 
	(namely,  node $\Delta i_2-1$ in Figure~1).

Let $m\in [\Delta N]$;	
	using Lemma~\ref{lem:permanentconnectivity}	and the above properties of the 
	graphs $G^{\Delta}(t)$, we inductively construct a path $m_{\ell}, \cdots, m_0$ in the graph $G^{\Delta}_0$
		starting at $m_{\ell}=m$, ending at $m_0 = \Delta j$, and such that each edge $(m_k,m_{k-1})$
		is in $E^{\Delta}(t_k)$ with 
		$t_k = t_{0} + k$.
Let us  denote $\theta_j (m)= t_{\ell}$;  by definition of the matrix $P(\theta_j(m))$, the existence of 
	this path is equivalent to $P_{m,\Delta j}(\theta_j(m))  >0$.
Finally, we let $$\theta_j = \max_{m\in [\Delta N]} \big (\theta_j (m) \big ) $$
	and we use Lemma~\ref{lem:SDeltaincrease} to complete the proof.
\end{proof}

We now give a condition  which ensures that the function $\pi_j$ does not decrease 
	when the set $S_j^{\Delta}$ remains stationnary.
	
\begin{lemma}\label{lem:outgoing}
If $S_j^{\Delta}(t) = S_j^{\Delta}(t+1)$ and $S_j(t)$ has  no outgoing edge in the graph $G(t+1)$, 
	then $\pi_j(t+1) \geq \pi_j(t)$.
\end{lemma}
\begin{proof}
Let $m \in S_j^{\Delta}(t)$, i.e., $P_{m, \Delta j}(t) >0$.
We consider two cases:
\begin{enumerate}
	\item $m$ is not a multiple of $\Delta$.
	By relation~R1, we have $$P_{m, \Delta j}(t+1) = P_{m+1, \Delta j}(t) \enspace.$$
	By Lemma~\ref{lem:SDeltaincrease}, we get $m \in S_j^{\Delta}(t+1)$, i.e., $P_{m, \Delta j}(t+1) >0$.
	Therefore, $m+1 \in S_j^{\Delta}(t)$ and $P_{m, \Delta j}(t+1) \geq \pi_j(t)$.
	\item $m$ is a multiple of $\Delta$, i.e., $m = \Delta i$ for some index $i\in [N]$.
	By relation~R2, we~have 
		$$P_{\Delta i, \Delta j}(t +1) = \sum_{k=1}^N A_{i, k}(t+1)\  P_{n_k, \Delta j}(t)
				\enspace,$$
	for some $n_k \in \{\Delta (k-1) +1, \dots, \Delta k \}$.
	It follows that 
	$$P_{\Delta i, \Delta j}(t +1) \geq 
		\sum_{k \in S_j(t)} A_{i, k}(t+1)\  P_{n_k, \Delta j}(t) \enspace.$$
	Using Lemma~\ref{lem:key} and the definition of $\pi_j(t)$, we obtain 
	$$P_{\Delta i, \Delta j}(t +1) \geq \pi_j(t) \times \sum_{k \in S_j(t)} A_{i, k}(t) \enspace.$$
	Since $i \in S_j(t)$ and $S_j(t)$ has  no outgoing edge in the graph~$G(t+1)$,
		we have $$ \sum_{k \in S_j(t)} A_{i, k}(t+1)  = \sum_{k=1}^N A_{i, k}(t+1) \enspace.$$
	The latter sum is equal to 1 as the matrix $A(t+1)$ is stochastic.
	Hence $P_{\Delta i, \Delta j}(t+1) \geq \pi_j(t)$.
\end{enumerate}
\end{proof}

We now put it all together to prove Theorem~\ref{thm:decent}.
Let $$\theta = \max_{j=1\cdots,N} \big(\theta_j \big)$$
	where the $\theta_j$'s are defined with regard to Proposition~\ref{pro:ergodic}.
Combining Lemma~\ref{lem:S0}, assumption~A3 with Lemmas~\ref{lem:phi},~\ref{lem:incoming}, and~\ref{lem:outgoing}
	we obtain 
	$$\forall t\geq \theta \ :\ \forall j\in [N] \ : \ 
	\pi_j(t) \geq \alpha^{\Delta N} \enspace.$$
By Corollary~\ref{cor:boundnorm}, we derive $ \lVert P(\theta) \lVert_{\bot} \leq 1 - N \alpha^{\Delta N}$.
In other words, we have shown that for any $t_0\geq \Delta -1$, there exists $\theta \geq t_0$
	such that 
	$$\lVert A^{\Delta}(\theta)\cdots  A^{\Delta}(t_0) \lVert_{\bot} \leq 1 - N \alpha^{\Delta N}\enspace.$$ 
Together with the sub-multiplicativity of the seminorm $\lVert  \cdot \lVert_{\bot}$,
	this implies that 
	$$\lim_{t\rightarrow +\infty} \lVert A^{\Delta}(t)\cdots  A^{\Delta}(0) \lVert_{\bot} =0 \enspace.$$
Theorem~\ref{thm:coord} follows from Proposition~\ref{pro:convergence} and the definition of
	 vector $X(t)$.

\section{Generalizations and remarks}

In this section, we present diverse strengthenings of Theorems~\ref{thm:coord}
 	and~\ref{thm:decent} which are obtained from direct generalizations of the arguments
	developed in the proofs of these theorems, or just simply by closely examining the proofs.
We conclude by a discussion of examples which demonstrate the role of the various
	assumptions in Theorems~\ref{thm:coord}--\ref{thm:decent*}.

\subsection{Eventual condition and coarser granularity}

Our proofs of theorems~\ref{thm:diamondC} and~\ref{thm:diamondD} are similar
	in the way we generalize the arguments in the proofs of Theorems~\ref{thm:coord}
 	and~\ref{thm:decent}, respectively.
We present only one of them, the proof of Theorem~\ref{thm:diamondD}.

Suppose that condition~$\Diamond$D holds for some time $T_0$ and some positive integer~$\Phi$.
We use the notation introduced in Section~\ref{sec:proofs}, Lemmas~\ref{lem:S0}--\ref{lem:incoming},
	and Lemma~\ref{lem:outgoing} for some time parameter $t_0\geq \max (T_0,\Delta -1)$.

Assume that $$S_j^{\Delta}(t) = \cdots = S_j^{\Delta}(t+\Phi) \enspace.$$
By definition of the sets $S_j^{\Delta}$ and $S_j$, all the sets
	$S_j(t),  \cdots , S_j (t+\Phi)$ are then equal to some set of nodes, which we denote by $S$.
Repeated application of Lemma~\ref{lem:incoming} show that 
	$S$ has no incoming edge 
	in each of the communication graphs~$G(t+1), \cdots, G(t+\Phi)$.
Hence $S$ has no incoming edge in~$H(t+1)$. 
Condition~$\Diamond$D then guarantees that $S$ has no outgoing edge in~$H(t+1)$.

Suppose now that $S$ has an outgoing edge $(i,k)$ in some communication graph
	$G(t+\varphi)$ with $\varphi \in [\, \Phi \,]$.
Because of the self-loop assumption~A2, we deduce that $(i,k)$ is an outgoing edge of~$S$
	in~$H(t+1)$, a contradiction.
Therefore, the set of nodes~$S$ has no outgoing edge in each of the communication graphs
	$G(t+1), \cdots, G(t+\Phi)$, and  Lemma~\ref{lem:outgoing} implies that
	$\pi(t+\Phi)\geq \cdots \geq \pi(t)$.
	
Using  the same arguments as for Theorem~\ref{thm:decent}, we conclude that there exists 
	$\theta \geq t_0$
	such that 
	$$\lVert A^{\Delta}(\theta)\cdots  A^{\Delta}(t_0) \lVert_{\bot} 
				\leq 1 - N\alpha^{\Phi \Delta N}\enspace.$$ 
Theorem~\ref{thm:diamondD} then follows from sub-multiplicativity of the seminorm 
	$\lVert  \cdot \lVert_{\bot}$, from Proposition~\ref{pro:convergence},
	and from the definition of $X(t)$.
				
\subsection{Partial complete reducibility}				
	
The proof of Theorem~\ref{thm:decent*} is based on the remark that 
	it suffices that one column 
	of a stochastic matrix~$A$ be positive to ensure that $A$ is contracting 
	(with respect to the seminorm~$\lVert \cdot \lVert_{\bot}$); 
	see Corollary~\ref{cor:boundnorm}.

Suppose that condition~D* holds for  some~$j_0\in [N]$.
We use the notation introduced in Section~\ref{sec:proofs}, Lemmas~\ref{lem:S0}--\ref{lem:incoming},
	and Lemma~\ref{lem:outgoing} for some time parameter $t_0\geq \Delta -1$
	and for node~$j_0$.
A close examination of the proof of Proposition~\ref{pro:ergodic}
	reveals that if the directed graph $([N]\,, \cup_{s\geq t_0}E(s))$ is 
	$j_0$-oriented (first part in condition~D*), then the $\Delta j_0$-th column of 
	the matrix~$P$ is eventually positive, i.e., there exists some time 
	$\theta_0 \geq t_0$ such that for  
		each 
		$$ \forall m\in [\Delta N]\ : \ P_{m,\Delta j_0}(\theta_0)  >0 \enspace.$$
Then Lemmas~\ref{lem:incoming} and~\ref{lem:outgoing} lead us to relax the complete reducibility
	property into the following property for a directed graph $G=(V,E)$, and a node $j_0\in V$.
\begin{description}
	\item[P$_{\mathbf{\mathit{j_0}}}$:]	There exists no subset of $V$ containing node~$j_0$ with 
	an outgoing edge and no incoming edge.
\end{description}

We now  study property~P$_{j_0}$ and  give
	an equivalent, but more tractable expression of it.
For that, we consider the {\em condensation $G^*$ of}~$G$ defined as the  directed  acyclic graph
	obtained by contracting each strongly connected component of~$G$ into a
	single node.
We denote by $i^*$ the strongly connected component of some node~$i$.

The next lemma, whose proof is obvious, allows us to restrict ourselves to the 
	case of acyclic graphs.
\begin{lemma}\label{lem:condensP}
A directed graph $G$ satisfies P$_{j_0}$ if and only if the condensation $G^*$ of~$G$
	satisfies P$_{j_0^*}$.
\end{lemma}

In turn, property~P$_{j_0}$ on a directed acyclic graph~$G$ admits a simple equivalent expression 
	in terms of the  connected components of~$G$.

\begin{lemma}\label{lem:sink}
If~$G$ is a  directed  acyclic graph, then~$G$ satisfies P$_{j_0}$ if and only if
 	(a)~node~$j_0$ is the  one and only sink of its own connected component, and~(b)
	 every other connected component of~$G$ reduces to a single isolated node.
\end{lemma}

\begin{proof}
For any node~$i$, let $I$, $I^-$, and $I^+$ denote  the connected component of~$i$,
 	the set of $i$'s ancestors, and the set of $i$'s descendants (both including~$i$), respectively.

First assume that (a) and (b) both hold.
Let $S$ be any subset of nodes with $j_0\in S$, and suppose that $(k_0,k_1)$ is an outgoing edge of $S$.
Then,  $k_0$ and $k_1$ are in the same connected component.
Condition~(b) implies that $k_1\in J_0$. 
Since $k_1 \notin S$, we have $k_1 \neq j_0$.
By condition~(a), $k_1$ is not a sink, and let $k_2$ be an outgoing neighbor of~$k_1$. 
If $k_2 \in S$, then $(k_1,k_2)$ is an incoming edge of~$S$;
	otherwise, we repeat the argument with $k_2$ instead of~$k_1$. 
In this way, we construct a sequence of nodes $k_1, k_2, \cdots$
	in the complement of~$S$.
Because $G$ is acyclic, this sequence is finite, i.e., $S$ has an incoming edge.

Conversely, suppose that either (a) or (b)  does not hold.
We consider the following three cases, and show that for each of them,
	$G$ does not satisfy P$_{j_0}$.
\begin{enumerate}
\item The node~$j_0$ is not a sink, i.e., $j_0$ has at least an outgoing neighbor~$k$.
Then, the set $J_0^-$ has no incoming edge, but an outgoing edge, namely $(j_0,k)$.
\item The node~$j_0$ is a sink and there is another sink, denoted $i$, in $J_0$.
Then,~$i$ is not an isolated node and has an incoming neighbor~$k$.
The complement of $\{i\}$ contains~$j_0$,  has no incoming edge, but an outgoing edge, namely $(k,i)$. 
\item There exists some edge $(i,k)$ with $i$ and $k$ both outside $J_0$.
Then, the set $J_0\cup  I^-$ has no incoming edge, but an outgoing edge, namely $(i,k)$.
\end{enumerate}
\end{proof}

As observed in~\cite{GB81} (Section~2, page 12), condition~(a) in Lemma~\ref{lem:sink}  can be expressed 
	in terms of $j_0$-orientation.
\begin{lemma}\label{lem:gb81}
Let $j_0$ be any node of a connected and directed acyclic graph~$G$.
Then, node~$j_0$ is the one and only sink of~$G$ if and only if $G$ is $j_0$-oriented.
\end{lemma}
We leave the simple proof of Lemma~\ref{lem:gb81} to the reader.

Finally, we compare orientation in a directed graph and in its condensation
	in the following lemma whose proof is trivial. 
\begin{lemma}\label{lem:condensor}
The directed graph $G$ is $j_0$-oriented if and only if the condensation of~$G$
		is $j_0^*\!$-oriented.
\end{lemma}

By combining the above four lemmas, we obtain an equivalent form of~P$_{j_0}$.		

\begin{proposition}\label{pro:Pj}
Let~$G$ be a directed  graph, and let~$j_0$ be any node of~$G$.
The  following two properties are equivalent.
\begin{enumerate}
	\item $G$ satisfies P$_{j_0}$.
	\item The connected component of $j_0$ in~$G$ is $j_0$-oriented, and 
	every other connected component  is strongly connected.
\end{enumerate}
\end{proposition}

\vspace{0.2cm}	

The end of the proof of Theorem~\ref{thm:decent*} is similar to the one of Theorem~\ref{thm:decent},
	except that the upper bound on the seminorm of matrix $P(\theta_0)$ is now
	$$\lVert A^{\Delta}(\theta_0)\cdots  A^{\Delta}(t_0) \lVert_{\bot} \leq 1 - \alpha^{\Delta N}\enspace.$$

\subsection{Examples} 

We now present two examples
	that demonstrate the roles of the self-loop 
	assumption A2 and of the conditions C, D. or D* in our convergence theorems.
In both we consider  the case of a synchronous system
		with 3 agents; in other words, $N=3$ and $\Delta =1$.
\vspace{0.1cm}		

In our first example, the sequence of matrices $(A(t))_{t\in \IN}$ is 3-periodic with
	$$A(0) = \left[ \begin{array}{ccc}
	              1 & 0 & 0  \\
	              0 & 0 & 1  \\
	              0 & 1 & 0 
	\end{array} 
	      \right],\   A(1) = \left[ \begin{array}{ccc}
		              0 & 0 & 1  \\
		              0 & 1 & 0  \\
		              1 & 0 & 0 
		\end{array} 
		      \right], \ 
			A(2) = \left[ \begin{array}{ccc}
			              0 & 1 & 0  \\
			              1 & 0 & 0  \\
			              0 & 0 & 1 
			\end{array} 
			      \right] \enspace. $$
The sequence of matrices  corresponds to the following scenario in the synchronous case.
Agent~1 communicates only with itself while agent~2 communicates with agent~3 and
	agent~3 communicates with agent~2, but none of the two agents~2 and~3 takes
	into account their own values.
This first round is then repeated infinitely often while rotating the communication graph.
We easily check that the algorithm actually keeps executing the  instruction $(x_1,x_2,x_3):= (x_3,x_2, x_1)$. 
Therefore, the algorithm does not achieve consensus, and does not even converge.
However, condition D and all the assumptions considered so far hold except 
	A2.
In fact, we even observe that the self-loops which occur at each node during 
every period of  duration 3 units of time
	do not help to achieve asymptotic consensus.
	
\vspace{0.3cm}	
We now briefly recall an example given in~\cite{BHOT05}, which
	shows how crucial is the fact that the agent~$j$ in conditions C or C* does not change over time.
Here, $x(0)=(0,1,0)$ and the  matrices in the sequence $(A(t))_{t\in \IN}$ are taken in the
	set $\{A_1,A_2,A_3\}$ where
	$$A_1 = \left[ \begin{array}{ccc}
	              1/2 & 0 & 1/2  \\
	              0 & 1 & 0  \\
	              0 & 0 & 1 
	\end{array} 
	      \right],\   A_2 = \left[ \begin{array}{ccc}
		              1 & 0 & 0  \\
		              0 & 1 & 0  \\
		              0 & 1/2 & 1/2 
		\end{array} 
		      \right], \ 
			A_3 = \left[ \begin{array}{ccc}
			              1 & 0 & 0  \\
			              1/2 & 1/2 & 0  \\
			              0 & 0 & 1 
			\end{array} 
			      \right] \enspace. $$
Given an increasing  sequence of  integers $(t_n)_{n\in \IN}$ with $t_0=0$ and to be chosen 
	suitably later,  we~let
	$$ A(t) = A_k $$
	when $ t_{3i +k-1} \leq t < t_{3i+k } $ for some non-negative integer~$i$.
This means, for instance, that 
	until time~$t_1$, agent~3 communicates with agent~1, and agent~1 forms the
	average of its own value and the value received from agent~3.
	
Let $(\epsilon_n)_{n\in \IN^*}$ be a sequence of positive reals  such
	that $\ell = \sum_{n=1}^{\infty} \epsilon_n <1/2$.
Times $t_1$ and $t_2$ are  chosen large enough to have $x_1(t_1) \geq 1-\epsilon_1$ 
	and $x_3(t_2) \leq \epsilon_1$.
Similarly, 	$t_3$ and $t_4$ are  chosen large enough to have $x_2(t_3) \geq 1-\epsilon_1 -\epsilon_2$ 
		and $x_1(t_4) \leq \epsilon_1 + \epsilon_2$, and so on.
The resulting vector $x(t)$ is such that each of its three entries is infinitely often
	at least equal to $1-\ell$ and infinitely often at most equal to $\ell$.
Since $\ell <1/2$, this proves that the sequence $(x(t))_{t\in \IN}$ is not convergent.

In this example, all the assumptions A, B, and D1 hold.
Moreover, the following weakening of~D* is satisfied:
 	at every time $t \in \IN$,  there is some $j \in [N]$ such that  the connected 
	component of $j$  in the communication graph $G(t)$ is $j$-oriented, and 
	every other connected component of~$G(t)$  is strongly connected.
Indeed, all the communication graphs have two connected components: one component 
	is reduced to 
	a single node with a self-loop, and the other one  is oriented with respect 
	to one single node (namely, node $k$ for the matrix $A_k$).

\section*{Acknowledgments}
I warmely thank Bernard Chazelle and Thomas Nowak for discussions and  helpful comments.

\bibliography{agents.bib}

\bibliographystyle{plain}

\end{document}